\documentclass[11pt,a4paper]{article}
\pdfoutput=1
\usepackage{amssymb,amsmath,amsthm,pdfsync}
\usepackage[top=2.5cm,right=3cm,left=3cm,bottom=2.5cm]{geometry}
\usepackage[english]{babel}
\usepackage{yhmath,url}

\allowdisplaybreaks[4]

\newtheorem{prop}{Proposition}
\newtheorem{thm}[prop]{Theorem}
\newtheorem{lemma}[prop]{Lemma}
\newtheorem{cor}[prop]{Corollary}

\newtheorem{rem}[prop]{Remark}

\newcommand{\pp}{\mathcal{P}}
\newcommand{\pa}{\mathcal{P}(\A)}
\newcommand{\sa}{\mathcal{S}(\A)}
\newcommand{\A}{\mathcal{A}}
\newcommand{\B}{\mathcal{B}}
\newcommand{\HH}{\mathcal{H}}
\newcommand{\R}{\mathbb{R}}
\newcommand{\C}{\mathbb{C}}
\newcommand{\Z}{\mathbb{Z}}
\newcommand{\de}{\mathrm{d}}
\newcommand{\inner}[1]{\left<\smash[t]{#1}\right>}
\newcommand{\mat}[1]{\bigg[\begin{array}{cc}#1\end{array}\bigg]}
\newcommand{\D}{D\mkern-11.5mu/\,}
\newcommand{\tr}{\mathrm{Tr}}
\newcommand{\norm}[1]{\|#1\|}
\newcommand{\ST}{C\hspace{-1pt}S}
\newcommand{\abs}[1]{|#1|}

\begin{document}
%%% ======================================================================

\title{On Pythagoras theorem for products of spectral triples}

\author{\rule{0pt}{15pt}Francesco D'Andrea$^1$ and Pierre Martinetti$\hspace{1pt}^{2,3}$ \\[15pt]
{\footnotesize $^1$ Dipartimento di Matematica e Applicazioni, Universit\`a di Napoli Federico II, Napoli, Italy.}
 \\[2pt]
{\footnotesize $^2$ Dipartimento di Matematica e CMTP, Universit\`a di Roma Tor Vergata, Roma, Italy.}
 \\[2pt]
{\footnotesize $^2$ Dipartimento di Fisica, Universit\`a di Roma Sapienza, Roma, Italy.}}

\date{}

\maketitle

{\renewcommand{\thefootnote}{}
\footnotetext{%
\hspace*{-6pt}\textit{MSC-class 2010:} 58B34 (Primary), 46L87 (Secondary). \\
\hspace*{10pt}\textit{Keywords:} Noncommutative geometry, spectral triples, spectral distance, Pythagoras theorem. \\
\hspace*{10pt}\textit{Acknowledgments.}~
\noindent P.M.~is supported by the ERG-Marie Curie fellowship 237927 ``NCG and quantum
gravity'' and the ERC Advanced Grant 227458 OACFT ``Operator Algebras and Conformal Field Theory''.}}

\vspace{-5pt}

\begin{abstract} We discuss a version of Pythagoras theorem in
noncommutative geometry. Usual
Pythagoras theorem can be formulated in terms of Connes' distance,  between pure states,  in the product
of commutative spectral triples. We investigate the
gene\-ralization to both non pure states and arbitrary spectral
triples. We show that Pythagoras theorem is replaced by
some Pythagoras inequalities, that we prove for the product of arbitrary
(i.e. non-necessarily commutative) spectral triples, assu\-ming only
some unitality condition. We show that these inequalities are optimal, and
provide non-unital counter-examples inspired by K-homology.
\end{abstract}

\smallskip

\section{Introduction}\label{sec:1}

\noindent Given the natural spectral triple $(\A, \HH, D)$ associated to a
complete Riemannian spin manifold $M$ without boundary, that is 
\begin{equation}
\A = C^\infty_0(M), \qquad\HH = L_2(M,S),\qquad D = \D,
\label{eq:23}
\end{equation}
the spectral distance $d_{\A,D}$ of Connes (see \S\ref{sec:3.1})
on the state space $\sa$ of $\A$ coincides
with the Wasserstein distance $W$ of order $1$ in the theory of
optimal transport with cost
function the geodesic distance $d_{\text{geo}}$. Namely, given two
probability measures $\mu, \mu'$ on $M$ viewed as states $\varphi,
\varphi'$ of $\A$, that is
$$
   \varphi(f) = \int_M f(x) d\mu \quad \forall f\in \A
$$
and similarily for $\varphi', \mu'$,  one has
{\begin{equation}
  \label{eq:21}
  d_{\A,D}(\varphi, \varphi') = W(\mu, \mu') \quad \forall \varphi, \varphi'\in\sa.
\end{equation}
On the space of pure states $\pa\simeq M$, that by Gelfand duality are Dirac's delta distributions
$
\delta_x(f) := f(x) \;\forall\; x\in M, f\in C(M),
$
the spectral/Wasserstein distance gives back the cost function
\begin{equation}
  \label{eq:14}
  d_{\A,D}(x,y) = W(\delta_x, \delta_y) =d_{\text{geo}}(x,y).
\end{equation}

Consider now the product of two manifolds $M_1$,
$M_2$ equipped with the product metric. Pythagoras theorem states that
\begin{equation}
  d_{\text{geo}}(x,x') = \sqrt{d_{\text{geo}}({x_1}, {x'_1})^2 + d_{\text{geo}}({x_2}, {x'_2})^2},
\label{eq:3}
\end{equation}
for any couple of points $x=(x_1, x_2)$ and $x'=(x'_1, x'_2)\in
M_1\times M_2$. Denoting $(\A, \HH, D)$ the product of
the spectral triples of $M_1$
and $M_2$, eq.~(\ref{eq:3}) can be equivalently formulated in terms
of spectral distances as
\begin{equation}
  \label{eq:15}
  d_{\A,D}(\delta_{x_1}\otimes \delta_{x_2}, \; \delta_{x'_1}\otimes
  \delta_{x'_2}) = \sqrt{d_{\A_1,D_1}(\delta_{x_1}, \delta_{x'_1})^2 + d_{\A_2,D_2}(\delta_{x_2}, \delta_{x'_2})^2}
\end{equation}
for any pairs of pure states
$\delta_{x_1}\otimes\delta_{x_2},
\delta_{x'_1}\otimes \delta_{x'_2}$ in $\pa$. 
In other terms, the
product of two manifolds (in the sense of spectral triple)  is
ortho\-gonal (in the sense of Pythagoras
theorem restricted to the pure state space). 

It is known for many
years \cite{Martinetti:2002ij}  that a similar result
 holds in the discrete case, that is for the product of a manifold by $\C^2$, as well as for the
 product of a manifold by the finite dimensional algebra
 $\C\oplus\mathbb{H}\oplus M_3(\C)$ describing the internal degrees of freedom of the
standard model of elementary particles
\cite{Chamseddine:2007oz}. Furthermore,  in
the last case Pythagoras theorem yields a metric interpretation of the Higgs field. Recently, it
comes out in \cite{MT11} that eq. (\ref{eq:15}) for
is also true
for the product of the Moyal plane by $\C^2$, but only between translated
states, that is for $\delta_{x'_1}, \delta_{x'_2}$ the two pure states
of $\C^2$, $\delta_{x_1}=\varphi$ any state of the Moyal
algebra and $\delta_{x_2}=\varphi\circ\tau_\kappa$ with $\tau_\kappa,
\kappa\in\R^2,$ the translation action of $\R^2$ on the Moyal
plane. For arbitrary two states of the Moyal algebra, it is not known whether Pythagoras
equality is still valid: a crucial tool of the proof that is missing
in the general case is the existence of a
geodesic between the states under consideration (as the Riemannian
geodesic on the manifold, and the orbit of the translation group on the Moyal plane).

In this paper, we
investigate the generalization of Pythagoras theorem to both non-pure
states and the product $(\A, \HH, D)$ of arbitrary spectral triples
$(\A_i, \HH_i, D_i)$, $i=1,2$. We only
impose two limitations: separable states, that is
$$\sa\ni \varphi :=\varphi_1\otimes\varphi_2,\qquad \varphi_i\in
{\mathcal S}(\A_i)$$ and
similarly for $\varphi'$, and unital spectral
triples. The restriction to separable states is natu\-ral with respect to the
commutative case, and is also discussed on some physical ground in
\cite[\S 2.2]{MMT11}. The restriction to unital spectral triples emerges from the computation, and is discussed in the last part of this
paper. 

Specifically, we show that the following
Pythagoras inequalities hold true on separable (non-necessarily
pure) states:\vspace{-0pt}
\begin{subequations}\label{eq:Pineq}
\begin{align}\label{eq:ineqSDa}
d_{\A,D}(\varphi,\varphi') &\geq
\sqrt{d_{\A_1,D_1}(\varphi_1,\varphi'_1)^2+d_{\A_2,D_2}(\varphi_2,\varphi'_2)^2} \;,\\[4pt]
d_{\A,D}(\varphi,\varphi') &\leq
\sqrt{2}\sqrt{d_{\A_1,D_1}(\varphi_1,\varphi'_1)^2+d_{\A_2,D_2}(\varphi_2,\varphi'_2)^2} \;.
\label{eq:ineqSDb}
\end{align}
\end{subequations}
In the non-unital case, only \eqref{eq:ineqSDb} holds true. 
As a corollary, one gets a Pythagoras ine\-quality  for the Wasserstein distance between
separable states of a product of manifolds:
\begin{equation}\label{eq:Wineq}
\sqrt{W_1(\mu_1,\mu_1')^2+W_2(\mu_2,\mu_2')^2}\leq
W(\mu,\mu')\leq \sqrt{2}\sqrt{W_1(\mu_1,\mu_1')^2+W_2(\mu_2,\mu_2')^2}.
\end{equation}
Moreover we show that both equations \eqref{eq:Pineq} and \eqref{eq:Wineq} are
optimal, i.e.~that the coefficient in \eqref{eq:ineqSDb} and on the r.h.s.~of
\eqref{eq:Wineq} cannot be less than $\sqrt{2}$, by providing examples where this bound
is actually saturated.

It is remarkable that ``something'' of Pythagoras theorem
survives in the most general case. This was not granted at all from
the beginning, especially having in mind the well known  ``inverse Pythagoras 
relation'' satisfied by the Dirac operator in the product of spectral
triples (see \eqref{eq:19} in the
conclusion). The later seems to indicate that Pythagoras equality for
the spectral distance  may be retrieved only in some very particular cases, like the product of a
manifold or the Moyal plane by $\C^2$,  when this inverse relation can
be inverted. It is rather unexpected that inequalities
(\ref{eq:Pineq}) holds true in the general case. 

Also, from the point of view of the Wasserstein distance and as far as we can judge from a limited knowledge of the 
literature, it did not seem so well noticed that Pythagoras theorem
does not hold for non pure states.

The paper is organized as follows. In \S\ref{sec:3.0} we
recall some basic definitions. In \S\ref{sec:2} we discuss Pythagoras theorem
for the product of two manifolds. We start with the geodesic distance (between
pure states): for the sake of completeness, we provide a
proof using differential geometry in \S\ref{sec:2.1}, and explain the relation
with noncommutative geometry and products of spectral triples in
\S\ref{sec:diffNCG}, proving eq.\eqref{eq:15}.
Then, in \S\ref{sec:2.2} we pass to non necessarily pure states and the Wasserstein distance,
and show with a simple counterexample that Pythagoras theorem does not hold, and must be
replaced by the inequality \eqref{eq:Wineq}; we prove that the latter is optimal. In \S\ref{sec:3.2} we consider the generalization to arbitrary spectral triples,
and prove the inequalities \eqref{eq:Pineq}: the upper bound for $d_{\A,D}(\varphi,\varphi')$
holds for arbitrary spectral triples, while the lower bound can be proved only for a product
of two \emph{unital} spectral triples. In \S\ref{sec:deg} we provide two simple examples of
non-unital spectral triples violating the above-mentioned lower bound, after a short digression
in \S\ref{sec:int} to explain their importance in $K$-homology.

\section{Spectral and Wasserstein distances}\label{sec:3.0}
%%% ======================================================================

We use the following notations/conventions.
$\B(\HH)$ is the algebra of all bounded linear operators on a Hilbert
space $\HH$.
By a $*$-algebra $\A$ we shall
always mean an associative involutive $\C$-algebra.
A $*$-representation $\pi:\A\to\B(\HH)$ is called \emph{non-degenerate} if $\{\pi(a)\psi\}_{a\in\A,\psi\in\HH}$ span a dense
subspace of $\HH$; when  $\A$ is unital, with unit element $e$, the representation $\pi$ is called \emph{unital} if
$\pi(e)=1\in\B(\HH)$ is the identity operator.

Note that a representation of a unital $*$-algebra is non-degenerate if and only if
it is unital. Indeed, since $e$ is a projection, $\pi(e)$ is a projection, and from
$\pi(a)=\pi(a)\pi(e)\;\forall\;a\in\A$ it follows that $\pi(\A)\HH$ coincides with
the range of $\pi(e)$; as a corollary, $\pi$ is non-degenerate if and only if $\ker\pi(e)=\{0\}$,
i.e.~if and only if $\pi(e)=1$.

When $\pi$ is a faithful representation, and there is no risk of ambiguity, 
we will identify $\A$ with $\pi(\A)$ and omit the representation symbol; if
$\exists\;e\in\A$ but $\pi(e)\neq 1$, we identify $\A$ with $\pi(\A)$ and think
of it as a non-unital subalgebra of $\B(\HH)$.

When we talk about \emph{states} of $\A$ we always mean states of the $C^*$-algebra
$\bar{\A}$, closure of $\pi(\A)$; the set of all states is denoted by $S(\A)$.
We denote by $\|\,.\,\|_{\B(\HH)}$ the operator
norm of $\B(\HH)$, by $\|v\|_{\HH}^2=\inner{v,v}$ the norm of a
vector $v\in\HH$, and use the notation $\inner{\,,\,}$ for the inner product,
regardeless of the Hilbert space we are considering.

\subsection{Basics on noncommutative spaces}\label{sec:3.1}
The core of noncommutative differential geometry is the notion of spectral triple \cite{Con94,Con95},
also known as \emph{unbounded Fredholm module} (see e.g.~\cite{Con85})
or \emph{K-cycle} (see e.g.~\cite{Con89}).

Recall that a \emph{spectral triple} $(\A,\HH,D)$ is the datum of:
i) a separable complex Hilbert space $\HH$,
ii) a $*$-algebra $\A$ with a faithful $*$-representation
$\pi:\A\to\B(\HH)$ (the representation symbol is usually omitted),
iii) a (not-necessarily bounded) self-adjoint operator $D$ on $\HH$
such that $[D,a]$ is bounded and $a(1+D^2)^{-1/2}$ is compact, for all $a\in\A$.
The spectral triple is \emph{unital} if $\A$ is a unital algebra and $\pi$ a unital
representation. In the last sections \S\ref{sec:int} and \S\ref{sec:4} we
will consider an example where the algebra is unital but the representation is not:
this will be regarded as a non-unital spectral triple. Note that
non-unital representations of unital algebras are of fundamental importance
in K-homology (see \S\ref{sec:int}).

A spectral triple is \emph{even} if there is a \emph{grading} $\gamma$ on $\HH$, i.e.~a bounded operator satisfying $\gamma=\gamma^*$ and $\gamma^2=1$,
commuting with any $a\in\A$ and anticommuting with $D$.

A commutative example is given by $(C^\infty_0(M),L^2(M,S),\D)$,
where $C_0^\infty(M)$ is the algebra of complex-valued smooth functions vanishing
at infinity on a Riemannian spin$^c$ manifold with no boundary, $L^2(M,S)$ is the
Hilbert space of square integrable spinors and $\D$ is the Dirac operator.
This spectral triple is even if $M$ is even dimensional.

The set $S(\A)$ is an extended metric space\footnote{%
An extended metric space is a pair $(X,d)$ where $X$ is a set
and $d:X\times X\to[0,\infty]$ a symmetric map satisfying the
triangle inequality and such that $d(x,y)=0$ if{}f $x=y$.
The only difference with an ordinary metric space is that the
value $+\infty$ for the distance is allowed.}, with distance given by
$$
d_{\A,D}(\varphi,\varphi')=\sup_{a=a^*\in\A}\big\{\varphi(a)-\varphi'(a)\,:\,
\|[D,a]\|_{\B(\HH)}\leq 1 \big\}
$$
for all $\varphi,\varphi'\in S(\A)$. This is usually called \emph{Connes' metric} or \emph{spectral distance}.
When there is no risk of ambiguity, this distance will be denoted
simply by $d$. It has been introduced in \cite{Con89}, with the
supremum on the whole algebra $\A$. It is routine to show that the
supremum can be equivalently searched on selfadjoint elements \cite{Iochum:2001fv}.

\medskip

Rieffel first noticed in \cite{Rieffel:1999ec} that the spectral distance
associated to \eqref{eq:23} for compact $M$ coincides with the Wasserstein distance of order $1$ between two
probability measures $\mu_1, \mu_2$ on $M$ (with cost given by
the geodesic distance $d_{\text{geo}}$) given by \cite{Vil08}:
$$
  W(\varphi_1, \varphi_2) := \sup_{\norm{f}_{\mathrm{Lip}}\leq 1,
    \,f\in L ^1(\mu_1)\cap L ^1(\mu_2) }\left(\int_{\cal X} f\de\mu_1 -
    \int_{\mathcal{X}} f\de\mu_2\right),
$$
where the supremum is on all real $\mu_{i=1,2}$-integrable functions $f$ that are $1$-Lipschitz,
i.e.~such that $\abs{f(x) -f(y)} \leq d_{\text{geo}}(x,y) \;\forall\; x,y \in M$.
This result remains true for
locally compact manifold providing one assumes geodesic completeness
(see \S2.2 of \cite{DM09}).

\subsection{Products of spectral triples}\label{sec:3}

In noncommutative geometry, the Cartesian product of spaces is replaced by the product
of spectral triples. Given two spectral triples
$(\A_1,\HH_1,D_1,\gamma_1)$ and $(\A_2,\HH_2,D_2)$ such that the first one is even, their product $(\A,\HH,D)$ is defined as
$$
\A=\A_1\otimes\A_2 \;,\qquad
\HH=\HH_1\otimes\HH_2 \;,\qquad
D=D_1\otimes 1+\gamma_1\otimes D_2 \;.
$$
Here the tensor product between algebras is the algebraic
tensor product. 

For simplicity of notations we will only consider the case when
at least one of the two spectral triples is even, but one can define the product
of two odd spectral triples as well (see e.g.~\cite{Vanhecke:1999uq,dabrodoss}),
and all our results can be extended to this case.

Recall that a state $\varphi:\A\to\C$ is called \emph{separable} if it is of the form
$
\varphi=\varphi_1\otimes\varphi_2 ,
$
with $\varphi_i$ a state on $\A_i$ for $i=1,2$.
When at least one of the $\A_i$ is commutative, all pure states are
separable \cite{Kadison1986}, that is
$
\pa = \pp (\A_1) \times \pp(\A_2) .
$

%%% ======================================================================
\section{Products of manifolds}\label{sec:2}

In this section, we first recall how to retrieve Pythagoras theorem
for a product $M = M_1\times M_2$ {of manifolds, and interpret this
easy result of differential geometry within the spectral
distance framework. Then we investigate the Wasserstein distance,
showing by examples that one cannot hope to prove inequalities stronger
than \eqref{eq:Wineq}.

\subsection{Pythagoras theorem: the differential geometry way}\label{sec:2.1}

Let  $(M_1, g_1)$, $(M_2,g_2)$ be two connected complete
Riemannian manifolds of dimension $m_1, m_2$, and let $M$ denote their
product $M_1\times M_2$ equipped with the product metric:
$$
g:=\left(\begin{array}{cc}  g_1& 0 \\ 0 &  g_2\end{array}\right).
$$
The line element
$\de s$ of $M$ is given by
$
\de s^2=\de s_1^2+\de s_2^2,
$
with $\de s_i$ the line element of $M_i$, $i=1,2$. 
This  infinitesimal version of Pythagoras theorem can be integrated
in order to obtain Pythagoras equality.

\begin{prop}\label{proppyth}
For any $x=(x_1, x_2), x'=(x'_1, x'_2) \in M$,
\begin{equation}
  \label{eq:4}
  d(x,x')^2 = d_1(x_1, x'_1)^2 + d_2(x_2, x'_2)^2
\end{equation}
where $d, d_i$ denote the geodesic distance on $M, M_i$ respectively, $i=1,2$. 
\end{prop}
\begin{proof}
Given a geodesic $c(s) = (c_1(s), c_2(s))$ between $x $ and $x'$ in
$M$ parametrized by its proper length $s$, we first show that the
projections $c_i(s)$ on $M_i$ satisfy the equation of the geodesics,
then that $s$ is an affine parameter for both curves $c_i$. Pythagoras
theorem follows immediately.
Let us compute the Christoffel symbol of $M$,
$$
  \Gamma_{ab}^c = \frac 12 g^{cd} \left( \partial_a g_{db}
    +  \partial_b g_{da} - \partial_d g_{ab}\right)
$$
where $g_{ab}$, $a,b\in[1, m_1+ m_2]$, denote the components
of the metric $g$ of $M$. Writing $g_{\mu\nu}, \mu,\nu\in\left[1,
  m_1\right]$ and $g_{\mu'\nu'}$, $\mu',\nu'\in\left[m_1+ 1,
  m_1 + m_2\right]$ the components of the metrics $g_1, g_2$ of $M_1, M_2$, one has 
$$
 g_{\mu\mu'} = g_{\mu'\mu}  = 0, \quad  \partial_\mu g_{\mu'\nu'}
 = \partial_{\mu'} g_{\mu\nu} = 0 \quad \forall\mu,\mu', \nu, \nu'
$$
so that for any $c\in[1,m_1 + m_2]$,
$
  \Gamma_{\mu\mu'}^c =\Gamma_{\mu'\mu}^c = \Gamma_{\mu\nu}^{\mu'}=
  \Gamma_{\mu'\nu'}^\mu = 0.
$
Therefore the geodesic equation on $M$:
$$
  \frac{d^2 c^c}{ds^2} + \Gamma_{ab}^c \frac{d c^a}{ds}
  \frac{dc^b}{ds} = 0
$$
separates into two equations on $M_1,M_2$:
\begin{equation}
\label{eqgeo}
  \frac{d^2 c_1^\alpha}{ds^2} + \Gamma_{\mu\nu}^{\alpha} \frac{d
    c_1^{\mu}}{ds} \frac{d c_1^\nu}{ds} = 0,\quad \quad
\frac{d^2 c_2^{\alpha'}}{ds^2} + \Gamma_{\mu'\nu'}^{\alpha'} \frac{d
  c_2^{\mu'}}{ds} \frac{d c_2^{\nu'}}{ds} = 0.
\end{equation}

Before claiming that $c_i$'s are geodesic curves in the $M_i$'s, one has
to check that $s$ is an affine parameter for both curves. To fix the
notations, we show it
for $M_1$, the proof for $M_2$ being similar. 
Let $s_1$ denote the proper length of the curve $c_1$ in $M_1$. Its length $l_1$ is
$$
\int_{0}^{l_i}\Big\|\frac{d}{ds_1} c_1\Big\| \,\de s_1.
$$
Under the change of parametrization $t\to s_1$, equation
\eqref{eqgeo} becomes
\begin{equation}
  \label{eq:5}
  \frac{d^2 c_1^\alpha}{ds_1^2} + \Gamma_{\mu\nu}^{\alpha} \frac{d
    c_1^{\mu}}{ds_1} \frac{d c_1^\nu}{ds_1} = - \frac{dc_1^\mu}{ds_1}
  \frac{d^2s_1}{ds^2}.
\end{equation}
The vector
  $\dot c_1 := \frac
  {dc_1^\mu}{ds_1}\partial_\mu$, tangent to $c_1$, has contant norm
  $\sqrt{g_1(\dot c_1 , \dot c_1)} = 1$. Hence, using that the metric
  is parallel with respect to the covariant derivative $\nabla_{\dot c_1}$
(of the Leci-Civita connection) along $\dot c_1$, that is
$$
\frac{d}{ds_1} g_1(\dot c_1 , \dot c_1) = 2 g_1(\nabla_{c_1} \dot c_1
,\dot c_1) = 0,
$$
one obtains from
\eqref{eq:5} --- whose l.h.s.~is nothing but $\nabla_{\dot c_1} \dot
c_1$ --- that:
$$
0= 2 g_1(\nabla_{c_1} \dot c_1
,\dot c_1) = -2 \frac{d^2s_1}{ds^2}g_1(\dot c_1, \dot c_1) = -2 \frac{d^2s_1}{ds^2}.
$$
Hence $s_1 = a_1s +b_1$ for some constants $a_1, b_1$. Similarly  $s_2= a_2 s + b_2$.
This means that $s$ is
an affine parameter of both curves $c_i$, so that the latter are
geodesics of $M_i$. 

One can parametrize any of the
curves $c, c_i$ by any of the parameters $s_i, s$. In particular,
using
$$
\frac{ds_1}{ds} = a_1, \quad \frac{ds_2}{ds} = a_2 \quad \text{ so that }\quad
\frac{ds_2}{ds_1} = \frac{a_2}{a_1},
$$
the length $l_2$ of $c_2$ can be written as
$$
  l_2 = \int_{0}^{l_1} \frac{ds_2}{ds_1} \de s_1 =
  \frac{a_2}{a_1}\int_{0}^{l_1} \de s_1 = \frac{l_1 a_2}{a_1}.
$$
As well, the length $l$ of $c$ is
\begin{align*}
  \int_0^{l_1} ds &= \int\left.\sqrt{1 +
      \left(\frac{ds_2}{ds_1}\right)^2}\right. ds_1 =
    \sqrt{1 + \left(\frac{a_1}{a_2}\right)^2}\int_0^{l_1} ds_1
    = l_1 \sqrt{1 + \left(\frac{a_1}{a_2}\right)^2} = \sqrt{l_1^2 + l_2^2},
\end{align*}
which is nothing but \eqref{eq:4}.
\end{proof}

\subsection{Pythagoras theorem: the noncommutative geometry way}\label{sec:diffNCG}

Besides the natural spectral triple (\ref{eq:23}), given an orientable
Riemannian manifold $M$ without boundary, one can define an even
spectral triple:
\begin{equation}
  \label{eq:17}
  \A=C^\infty_0(M), \qquad \HH=\Omega^\bullet(M), \qquad
  D=\de+\de^*,
\end{equation}
with $\HH$ the Hilbert space of square integrable differential
forms and $D$ the Hodge-Dirac operator (self-adjoint on
a suitable domain).
The grading $\gamma\omega:=(-1)^k\omega$ on \mbox{$k$-form} is
extended by linearity on $\HH$.
This spectral triple is even, even if $M$ is odd-dimensional.
We will refer to this as the ``canonical spectral triple'' of $M$,
and denote it by $\ST(M)$.

If $M$ is even-dimensional, there are in fact two possible $\Z_2$-gradings on $\HH$,
both anticommuting with $D$: one is the grading $\gamma$ above, and the other is the
Hodge star operator (whose square is $1$ using the phase convention of
\cite{GVF01}). Therefore, one has two
spectral triples that differ only in the grading (thus giving the same distance): in the
former case, one usually calls $D$ the \emph{Hodge-Dirac operator} and the corresponding
differential complex is the de Rham complex; in the latter case, the operator $D$ is usually
called the \emph{signature operator} and the corresponding differential complex is called the
signature complex. The signature operator is the one used in Connes reconstruction formula \cite{Con08},
and it is the one interesting in K-homology and index theory \cite{Gil96}. On the other hand,
the de Rham complex is the one which is multiplicative under products, as explained in \S3.1 of \cite{Gil96}
and recalled below.
Working with the Hodge-Dirac operator has the advantage (besides the fact that we don't need
a spin structure), that one can use the product of even-even spectral triples, even if both
the manifolds $M_1$ and $M_2$ are odd-dimensional.

Let now $M=M_1\times M_2$ be the product of two orientable
Riemannian manifolds $M_1, M_2$ (with product
metric). Identifying $C_0(M)$ with the spatial tensor product $C_0(M_1)\bar\otimes C_0(M_2)$ \cite[App.~T]{WO93},
all pure states of $C_0(M)$ are separable:
$\delta_x=\delta_{x_1}\otimes\delta_{x_2}\; \forall x=(x_1,x_2)\in M$.

\begin{prop}
  \label{eq:18}
  The spectral distance associated to the spectral triple $\ST(M)$ coincides with the spectral
  distance associated to the product of the canonical spectral triples
  of $M_1$ and $M_2$, that we denote by $\ST(M_1)\otimes\ST(M_2)$.
\end{prop}
\begin{proof}
If $\omega_1$ resp.~$\omega_2$ is a differential form on $M_1$ resp.~$M_2$ (with degree $k_1$ resp.~$k_2$),
there is an obvious identification $\Omega^\bullet(M)\simeq\Omega^\bullet(M_1)\otimes\Omega^\bullet(M_2)$
given by the map
\mbox{$
\omega_1\wedge\omega_2\to \omega_1\otimes\omega_2
$}
(the linear span of forms $\omega_1\wedge\omega_2$ is dense in $\Omega^\bullet(M)$), and by the graded Leibniz rule
$
\de(\omega_1\wedge\omega_2)=(\de\omega_1)\wedge\omega_2+(-1)^{k_1}\omega_1\wedge (\de\omega_2) ,
$
that is
$$
\de=\de|_{M_1}\otimes 1+\gamma_1\otimes \de|_{M_2} \;.
$$
By adjunction, one has a similar relation for $\de^*$, proving that
the Hodge-Dirac operator on $M$ is $D=D_1\otimes 1+\gamma_1\otimes D_2$,
where $D_i$ is the Hodge-Dirac operator on $M_i$.
Since the degree of $\omega_1\wedge\omega_2$ is the sum of the degrees of $\omega_1$ and $\omega_2$, one has also
$\gamma=\gamma_1\otimes\gamma_2$. In other terms, the Dirac operator
and chirality of $\ST(M)$ are the Dirac operator and grading of the product $\ST(M_1)\otimes\ST(M_2)$.

However $\ST(M)$ is not equal nor unitary equivalent to $\ST(M_1)\otimes\ST(M_2)$, since
the algebraic tensor product \mbox{$\A_1\otimes\A_2=C_0^\infty(M_1)\otimes C_0^\infty(M_2)$} is only dense in
the algebra:
$$
\A=C_0^\infty(M_1\times M_2)\simeq C_0^\infty(M_1)\,\hat\otimes\, C_0^\infty(M_2) \;,
$$
where $\hat\otimes$ is the projective tensor product of complete locally convex Hausdorff topological algebras \cite{Gro79}.
Since $\A_1\otimes\A_2\subset\A$, clearly
$
d_{\A_1\otimes\A_2,D}(\varphi,\varphi')\leq d_{\A,D}(\varphi,\varphi') .
$

In fact, the two distances coincide. By definition of projective
tensor product, the topology of $\A$ is given by the uniform convergence of functions
together with all their derivatives: every element $f\in\A$ is the limit
of a sequence of elements $f_n\in\A_1\otimes\A_2$ which is convergent in the above-mentioned topology.
In particular, since $f_n$ is
norm-convergent to $f$, one has $\varphi(f_n)\to\varphi(f)$ for any state $\varphi$.
Moreover, the uniform convergence coincides with the convergence in the sup norm (that is also the operator norm
on $\HH$), so $[D,\pi(f_n)]$ is also norm-convergent to $[D,\pi(f)]$. This proves that 
$
d_{\A_1\otimes\A_2,D}(\varphi,\varphi')=d_{\A,D}(\varphi,\varphi') .
$
\end{proof}

\begin{rem}
Up to a completion of the algebra $\A_1\otimes \A_2$, the spectral triples $\ST(M)$ and $\ST(M_1)\otimes \ST(M_2)$
are equivalent. In particular, the Hodge-Dirac operator of $M$ is related to the Hodge-Dirac
operators of $M_1$ and $M_2$ by the formula $D=D_1\otimes 1+\gamma_1\otimes D_2$.
A similar decomposition for the Dirac's Dirac operator
holds for $\R^n$ and flat tori \cite{dabrodoss}, and is believed to be true for arbitrary
Riemannian spin manifolds.
\end{rem}

Prop.~\ref{eq:18} applied to pure states, together with
Prop.~\ref{proppyth}, shows that the spectral
distance associated to $\ST(M_1)\otimes\ST(M_2)$ is the geodesic distance
of $M_1\times M_2$. In other terms, the product of canonical
spectral triples of manifolds is orthogonal in the sense of
\eqref{eq:15}.

%%% ======================================================================

\subsection{Pythagoras inequalities for the Wasserstein distance}\label{sec:2.2}

Pythagoras theorem holds true for pure states in the product of
commutative spectral triples. There are two possible
generalization: non-pure states and noncommutative spectral triples. We show on elementary examples that even in the commutative case,
Pythagoras theorem does not hold for non-pure states. Noncommutative examples
are investigated in the next section.

Consider the Cartesian product $\R\times\R$ with the standard Euclidean
metric, and the states
\begin{equation}\label{eq:omegalambda}
\varphi_\lambda(f):=\lambda f(1)+(1-\lambda)f(0) \;,
\end{equation}
with $0\leq \lambda\leq 1$. Let us denote by $W_1$ resp.~$W_2$ the
Wasserstein distance on the first resp.~second factor of $\R\times\R$,
and by $W$ the Wasserstein distance on the product.

\begin{prop}\label{prop:Rsquare}
Let $k_\lambda=\lambda+\sqrt{2}(1-\lambda)$. Then:
$$
W(\varphi_\lambda\otimes\varphi_\lambda,\varphi_0\otimes\varphi_0)=k_\lambda\sqrt{W_1(\varphi_\lambda,\varphi_0)^2+W_2(\varphi_\lambda,\varphi_0)^2} \;,
$$
for any $0\leq\lambda\leq 1$. Note that $k_\lambda$ assumes all possible values in $[1,\sqrt{2}]$.
\end{prop}

\begin{proof}
As recalled in \S\ref{sec:3.1},
$W_1(\varphi_\lambda,\varphi_0)=W_2(\varphi_\lambda,\varphi_0)$ is the supremum of
$\lambda\{f(1)-f(0)\}$ over real $1$-Lipschitz functions $f$ on $\R$. This is equal
to $\lambda$ (the sup is attained on the function $f(x)=x$).

On the other hand, identifying $\sum_if_i\otimes g_i\in C_0(\R)\otimes
C_0(\R)$ with the function $h\in C_0(\R^2)$, $h(x_1,x_2)=\sum_if_i(x) g_i(y)$,
one has 
$$
  (\varphi_\lambda \otimes \varphi_\lambda)(h) = \lambda^2 h(1,1) +
  \lambda(1-\lambda)(h(1,0) + h(0,1)) + (1-\lambda)^2 h(0,0).  
$$
Therefore $W(\varphi_\lambda\otimes\varphi_\lambda,\varphi_0\otimes\varphi_0)$
is the supremum of
\begin{equation}\label{eq:fonRtwo}
\lambda^2\{h(1,1)-h(0,0)\}+\lambda(1-\lambda)\{h(1,0)-h(0,0)\}
+\lambda(1-\lambda)\{h(0,1)-h(0,0)\} \;,
\end{equation}
over $1$-Lipschitz functions $h$ on $\R^2$. From $h(x)-h(y)\leq d_{\text{geo}}(x,y)$
it follows that this is no greater than
$\sqrt{2}\lambda^2+2\lambda(1-\lambda)=\sqrt{2}\lambda k_\lambda$.
The supremum is saturated by the function $h(x_1,x_2)=\sqrt{x_1^2+x_2^2}$,
proving that
\begin{equation*}
W(\varphi_\lambda\otimes\varphi_\lambda,\varphi_0\otimes\varphi_0)
=\sqrt{2}\lambda k_\lambda = k_\lambda {\sqrt{W_1(\varphi_\lambda,\varphi_0)^2+W_2(\varphi_\lambda,\varphi_0)^2}}.
\vspace{-24pt}
\end{equation*}
\end{proof}
\noindent Note that for $\lambda\to 0^+$, $k_\lambda$ goes to $\sqrt{2}$ and not to zero, although
for $\lambda=0$, Pythagoras equality is trivially satisfied.

\smallskip

One can show that the same argument works on a torus (with
flat metric), providing then an example where the space is a compact one.
These examples show that the best one may hope, for arbitrary
states and manifolds, is an inequality like \eqref{eq:Wineq}. In the next section, we prove
such an inequality, also holding in the noncommutative case.

\section{Pythagoras inequalities for products of spectral triples}\label{sec:3.2}

In this section we consider a product of arbitrary (not necessarily commutative)
spectral triples. We shall use the shorthand notation $d=d_{\A,D}$ and $d_i=d_{\A_i,D_i}$ for $i=1,2$.
Let us state the main theorem.

\begin{thm}\label{thm1}
Given two spectral triples $(\A_i, \HH_i, D_i)$, $i=1,2$, one has:
\begin{itemize}\itemsep=-5pt
\item[i)] For any two separable states $\varphi=\varphi_1\otimes\varphi_2$ and $\varphi'=\varphi_1'\otimes\varphi_2'$, we have
\begin{subequations}\label{eq:main}
\begin{equation}
d(\varphi,\varphi') \leq d_1(\varphi_1,\varphi'_1)+d_2(\varphi_2,\varphi'_2) \;.
\label{eq:mainA}
\end{equation}
\item[ii)] Furthermore, if the spectral triples are unital, we also have:
\begin{equation}
d(\varphi,\varphi') \geq \sqrt{d_1(\varphi_1,\varphi'_1)^2+d_2(\varphi_2,\varphi'_2)^2} \;.
\label{eq:mainB}
\end{equation}
\end{subequations}
\end{itemize}
\end{thm}

\noindent
Notice that from \eqref{eq:mainA} and the observation that
$(a+b)^2=2(a^2+b^2)-(a-b)^2\leq 2(a^2+b^2)$ it follows
\begin{equation}\label{eq:sqrtoftwo}
d(\varphi,\varphi')\leq
\sqrt{2}\sqrt{d_1(\varphi_1,\varphi'_1)^2+d_2(\varphi_2,\varphi'_2)^2} \;.
\end{equation}
As an easy corollary, one also retrieves a result of \cite{Martinetti:2002ij}: 
\begin{cor}\label{cor9}
In the unital case, if $\varphi_2=\varphi'_2$ we have
$
d(\varphi,\varphi')=d_1(\varphi_1,\varphi'_1) ,
$
and similarly if $\varphi_1=\varphi'_1$ we have
$
d(\varphi,\varphi')=d_2(\varphi_2,\varphi'_2) .
$
\end{cor}

\noindent
Theorem \ref{thm1} generalizes to arbitrary spectral triples
the results of \cite{MT11}, where the first triple was assumed to be unital,
and the second one was the canonical spectral triple on $\C^2$.

Let us recall that by unital spectral triple we mean that both the conditions \mbox{$\exists\,e\in\A$}
and $\pi(e)=1$ are satisfied. The importance of this requirement for (\ref{eq:mainB}) is
discussed in \S\ref{sec:deg}.
If $\exists\;e\in\A$ but $\pi(e)\neq 1$, we still have a legitimate spectral triple, although
non-unital, and we stress that \eqref{eq:mainA} is still valid in this case.

\subsection{Proof of the main theorem}\label{sec:3.3}
This section is devoted to the proof of theorem \ref{thm1}.
We need some preliminary lemmas.

\begin{lemma}\label{lemma2}
For any $x,y\geq 0$,
\begin{equation}\label{eq:trick}
\sup_{\substack{\alpha,\beta\geq 0 \\[1pt] \alpha^2+\beta^2\leq 1}}
(\alpha x+\beta y)=\sqrt{x^2+y^2} \;.
\end{equation}
\end{lemma}
\begin{proof}
If $(x,y)=(0,0)$
the statement is trivial. Assuming $(x,y)\neq (0,0)$, by choosing
$(\alpha,\beta)=(x,y)/\sqrt{x^2+y^2}$ one proves that the left
hand side of \eqref{eq:trick} is greater than or equal to the
right hand side. On the other hand, by Cauchy-Bunyakovsky-Schwarz inequality
$$
\alpha x+\beta y\leq
\sqrt{\alpha^2+\beta^2} 
\sqrt{x^2+y^2}\leq
\sqrt{x^2+y^2}
$$
if $\alpha^2+\beta^2\leq 1$. This proves that the inequality is
actually an equality.
\end{proof}

\begin{lemma}\label{lemma3}
For any $a=a_1\otimes 1+1\otimes a_2$, with $a_i\in\A_i$, we have
$$
\|[D,a]\|_{\B(\HH)}^2=\|[D_1,a_1]\|_{\B(\HH_1)}^2+\|[D_2,a_2]\|_{\B(\HH_2)}^2 \;.
$$
\end{lemma}
\begin{proof}
We have
$$
[D,a]=[D_1,a_1]\otimes 1+\gamma_1\otimes [D_2,a_2] \;.
$$
Call $A=[D_1,a_1]\otimes 1$, $B=1\otimes [D_2,a_2]$,
$\gamma=\gamma_1\otimes 1$, and notice that
$\|A\|_{\B(\HH)}=\|[D_1,a_1]\|_{\B(\HH_1)}$,
$\|B\|_{\B(\HH)}=\|[D_2,a_2]\|_{\B(\HH_2)}$,
$\|A+\gamma B\|_{\B(\HH)}=\|[D,a]\|_{\B(\HH)}$.
The lemma amounts to prove that
$$
\|A+\gamma B\|_{\B(\HH)}^2=\|A\|_{\B(\HH)}^2+\|B\|_{\B(\HH)}^2 \;.
$$
Since $A=-A^*$, $B=-B^*$,
$A\gamma+\gamma A=0$, $B\gamma-\gamma B=0$ and $[A,B]=0$,
we have
$$
(A+\gamma B)^*(A+\gamma B)=-A^2-B^2+\gamma[A,B]=-A^2-B^2 \;,
$$
so that, by the triangle inequality and the $C^*$-norm property,
$$\|A+\gamma B\|_{\B(\HH)}^2\leq\|A^2\|_{\B(\HH)}+\|B^2\|_{\B(\HH)}.$$
To prove the opposite
inequality, let us consider 
a supremum
over vectors $v=v_1\otimes v_2$:
\begin{align*}
\|A+\gamma B\|_{\B(\HH)}^2 &\geq\sup_{0\neq v=v_1\otimes
  v_2\in\HH}\frac{\inner{(A+\gamma B)v,(A+\gamma B) v}}{\|v\|_{\HH}^2}
&=&\sup_{0\neq v=v_1\otimes v_2\in\HH}\frac{\inner{v,-(A^2+B^2)v}}{\|v\|_{\HH}^2}\\
&=\sup_{0\neq v=v_1\otimes v_2\in\HH}\left\{
\frac{\inner{v,-A^2v}}{\|v\|_{\HH}^2}+\frac{\inner{v,-B^2v}}{\|v\|_{\HH}^2}\right\}
&=&\sup_{0\neq v=v_1\otimes v_2\in\HH}\left\{
\frac{\inner{Av,Av}}{\|v\|_{\HH}^2}+\frac{\inner{Bv,Bv}}{\|v\|_{\HH}^2}\right\}\\
&=\|A\|^2_{\B(\HH_1)}+\|B\|^2_{\B(\HH_2)} \;.& &
\end{align*}
This proves the lemma.
\end{proof}

\noindent
\textbf{Proof of Theorem \ref{thm1}, point (ii).}
Let $(\A_1,\HH_1,D_1,\gamma_1)$ and $(\A_2,\HH_2,D_2)$ be two unital
spectral triples and $\varphi=\varphi_1\otimes\varphi_2$ and $\varphi'=\varphi_1'\otimes\varphi_2'$
two  separable states. By definition
$$
d(\varphi,\varphi')
=\sup_{a=a^*\in\A}
\Big\{ \varphi(a)-\varphi'(a)\;:\;\|[D,a]\|^2_{\B(\HH)}\leq 1 \Big\} \;.
$$
We get a lower bound if we take the supremum over elements of
the form $a=a_1\otimes 1+1\otimes a_2$, with $a_1=a_1^*\in\A_1$
and $a_2=a_2^*\in\A_2$. Since
$$
\varphi(a)-\varphi'(a)=\varphi_1(a_1)-\varphi_1'(a_1)+\varphi_2(a_2)-\varphi_2'(a_2) \;,
$$
by Lemma \ref{lemma3} we get
\begin{align*}
d(\varphi,\varphi') &\geq
 \sup_{a_i=a^*_i\in\A_i}
\Big\{ \varphi_1(a_1)-\varphi_1'(a_1)+\varphi_2(a_2)-\varphi_2'(a_2) \,:\\
& \hspace{5cm}
\|[D_1,a_1]\|_{\B(\HH_1)}^2+\|[D_2,a_2]\|_{\B(\HH_2)}^2 \leq 1 \Big\}
\\
&=\sup_{\alpha^2+\beta^2\leq 1}
\left\{
\sup_{a_1=a^*_1\in\A_1}
\Big\{ \varphi_1(a_1)-\varphi_1'(a_1)\,:\,\|[D_1,a_1]\|_{\B(\HH_1)}\leq\alpha \Big\}
\right.
\\
&\hspace{3cm} +
\left.\sup_{a_2=a^*_2\in\A_2}
\Big\{ \varphi_2(a_2)-\varphi_2'(a_2)\,:\,\|[D_2,a_2]\|_{\B(\HH_2)}\leq\beta \Big\}
\right\}
\\
&=\sup_{\alpha^2+\beta^2\leq 1}
\left\{\alpha\,d_1(\varphi_1,\varphi'_1)+\beta\,d_2(\varphi_2,\varphi'_2)
\right\} \;.
\end{align*}
Applying Lemma \ref{lemma2} to last equation, we prove \eqref{eq:mainB}.
\qed

\medskip

\noindent
{\bf Remark:} note that what spoils the proof
in the non-unital case is the fact that operators of the form
$a=a_1\otimes 1+1\otimes a_2$ are in general not elements of
$\A=\A_1\otimes\A_2$.

\begin{lemma}\label{lemma:pos}
Let $a\in\A=\A_1\otimes \A_2$ and for any two $\varphi_i\in S(\A_i)$,
$i=1,2$, let us call $a_1:=(\mathtt{id}\otimes\varphi_2)(a)\in\A_1$
and
$a_2:=(\varphi_1\otimes \mathtt{id})(a)\in\A_2$. Then
$$
\|[D_1,a_1]\|_{\B(\HH_1)}\leq\|[D_1\otimes 1,a]\|_{\B(\HH)}
\qquad\text{and}\qquad
\|[D_2,a_2]\|_{\B(\HH_2)}\leq\|[\gamma_1\otimes D_2,a]\|_{\B(\HH)}
\;.
$$
\end{lemma}
\begin{proof}
We use the obvious identification of $\A_1\otimes\C\,1$ with $\A_1$
and $\C\,1\otimes\A_2$ with $\A_2$. We also identify $\B(\HH_1)\otimes\C\,1$
with $\B(\HH_1)$ and $\C\,1\otimes\B(\HH_2)$ with $\B(\HH_2)$.

States have norm $1$. Hence the maps
$\mathtt{id}\otimes\varphi_2:\B(\HH_1)\otimes\A_2\to\B(\HH_1)$
and
$\varphi_1\otimes \mathtt{id}:\A_1\otimes\B(\HH_2)\to\B(\HH_2)$
have norm $1$. This means
\begin{align*}
\|[D_1,a_1]\|_{\B(\HH_1)}
&=\|[D_1,(\mathtt{id}\otimes\varphi_2)(a)]\|_{\B(\HH_1)}
\\
&=\|(\mathtt{id}\otimes\varphi_2)[D_1\otimes 1,a]\|_{\B(\HH_1)}
\leq\|[D_1\otimes 1,a]\|_{\B(\HH)} \;.
\end{align*}
Similarly, one has $\|[D_2,a_2]\|_{\B(\HH_2)}\leq\|[1\otimes D_2,a]\|_{\B(\HH)}$.
To conclude the proof we notice that, since $\gamma_1^*\gamma_1=1$,
$
\|[1\otimes D_2,a]\|_{\B(\HH)}=\|(\gamma_1\otimes 1)[1\otimes D_2,a]\|_{\B(\HH)}
=\|[\gamma_1\otimes D_2,a]\|_{\B(\HH)} .
$
\end{proof}

\begin{lemma}\label{lemma:sqrt2}
Let $\gamma$ be a grading, $A$ an odd operator and $B$ an even operator
(i.e.~$A\gamma+\gamma A=0$ and $B\gamma-\gamma B=0$).
Then
\begin{equation}\label{eq:ineqA}
\max\big\{\|A\|_{\B(\HH)},\|B\|_{\B(\HH)}\big\}\leq\|A+B\|_{\B(\HH)} \;.
\end{equation}
\end{lemma}

\begin{proof}
From the triangle inequality we obtain
$$
\|A\|_{\B(\HH)}=
\left\|\tfrac{A+B}{2}+\tfrac{A-B}{2}\right\|_{\B(\HH)}
\leq \tfrac{1}{2}\|A+B\|_{\B(\HH)}+\tfrac{1}{2}\|A-B\|_{\B(\HH)} \;,
$$
and
$$
\|B\|_{\B(\HH)}=
\left\|\tfrac{A+B}{2}-\tfrac{A-B}{2}\right\|_{\B(\HH)}
\leq \tfrac{1}{2}\|A+B\|_{\B(\HH)}+\tfrac{1}{2}\|A-B\|_{\B(\HH)} \;.
$$
But $A-B=-\gamma(A+B)\gamma$ with $\gamma=\gamma^*$
unitary, so $\|A-B\|_{\B(\HH)}=\|A+B\|_{\B(\HH)}$. This proves
that
$\|A\|_{\B(\HH)}\leq \|A+B\|_{\B(\HH)}$
and
$\|B\|_{\B(\HH)}\leq \|A+B\|_{\B(\HH)}$,
i.e.~the inequality \eqref{eq:ineqA}.
\end{proof}

\begin{cor}\label{cor:ineq}
For any $a\in\A$, we have
\begin{equation}\label{eq:ineqC}
\max\big\{
\|[D_1\otimes 1,a]\|_{\B(\HH)},\|[\gamma_1\otimes D_2,a]\|_{\B(\HH)}\big\}\leq \|[D,a]\|_{\B(\HH)}^2 \;.
\end{equation}
\end{cor}

\begin{proof}
Apply  Lemma \ref{lemma:sqrt2} with $A=[D_1\otimes 1,a]$, $B=[\gamma_1\otimes D_2,a]$
and $\gamma=\gamma_1\otimes 1$.
\end{proof}

\smallskip

\noindent
\textbf{Proof of Theorem \ref{thm1}, point (i).}
Let $\varphi=\varphi_1\otimes\varphi_2$ and
$\varphi'=\varphi_1'\otimes\varphi_2'$ be
two separable states.
Any $a\in\A$ can be written as $a=\sum_ia_1^i\otimes a_2^i$.
Notice that
\begin{align*}
\varphi(a)-\varphi'(a)
&=\sum\nolimits_i\varphi_1(a_1^i)\varphi_2(a_2^i)-\varphi_1'(a_1^i)\varphi_2'(a_2^i)
\notag\\
&=\sum\nolimits_i\big\{\varphi_1(a_1^i)-\varphi_1'(a_1^i)\big\}\varphi_2(a_2^i)
+\varphi_1'(a_1^i)\big\{\varphi_2(a_2^i)-\varphi_2'(a_2^i)\big\}
\notag\\
&=\varphi_1(a_1)-\varphi_1'(a_1)+\varphi_2(a_2)-\varphi_2'(a_2)
\notag\\
&\leq d_1(\varphi_1,\varphi_1')\|[D_1,a_1]\|_{\B(\HH_1)}
+d_2(\varphi_2,\varphi_2')\|[D_2,a_2]\|_{\B(\HH_2)} \;,
\end{align*}
where we used the linearity of states and called
\begin{equation*}
a_1=\sum\nolimits_ia_1^i\varphi_2(a_2^i)\in\A_1 \;,\qquad
a_2=\sum\nolimits_ia_2^i\varphi_1'(a_1^i)\in\A_2 \;.
\end{equation*}
Using Lemma \ref{lemma:pos} we deduce that
$$
\varphi(a)-\varphi'(a)
\leq d_1(\varphi_1,\varphi_1')\|[D_1\otimes 1,a]\|_{\B(\HH)}
+d_2(\varphi_2,\varphi_2')\|[\gamma_1\otimes D_2,a]\|_{\B(\HH)} \;.
$$
By \eqref{eq:ineqC} we get
$$
\varphi(a)-\varphi'(a)
\leq \big\{d_1(\varphi_1,\varphi_1')
+d_2(\varphi_2,\varphi_2')\big\}\|[D,a]\|_{\B(\HH)} \;,
$$
and taking the sup over $\A$ with $\|[D,a]\|_{\B(\HH)}\leq 1$
we get \eqref{eq:mainA}.
\qed

\bigskip

We remark that, unlike \eqref{eq:mainB}, \eqref{eq:mainA} and \eqref{eq:sqrtoftwo}
are valid for arbitrary (not necessarily unital) spectral triples. In
\S\ref{sec:4} and \S\ref{sec:5} we give two counterexamples to \eqref{eq:mainB}
using non-unital spectral triples. In the next section
we make a short digression to explain the importance of this counterexamples,
arising in the study of K-homology.

%%% ======================================================================
\section{Interlude on the one-point and two-point spaces}\label{sec:int}

\subsection{K-homology of $\C$}

An even pre-Fredholm module $(\A,\HH,F,\gamma)$ over a $*$-algebra $\A$ is given by a $\Z_2$-graded Hilbert
space $\HH$ with grading $\gamma$, a representation $\pi:\A\to\B(\HH)$, commuting with the grading, 
and a bounded operator $F$ anticommuting with the grading, such that
\mbox{$\pi(a)(F-F^*)$}, $\pi(a)(F^2-1)$ and $[F,\pi(a)]$ are compact operators for all $a\in\A$%
\footnote{Here we adopt the terminology of \cite[\S8.2]{GVF01}.
In \cite{BJ83,HR00} pre-Fredholm modules are called Fredholm modules `tout court'.}.
With a suitable equivalence relation, classes of even
pre-Fredholm modules form the zeroth K-homology group $K^0(\A)$, see e.g.~\cite[\S8.2]{HR00}.

Given a spectral triple $(\A,\HH,D,\gamma)$, a pre-Fredholm module $(\A,\HH,F,\gamma)$ can be obtained
by replacing $D$ with $F=D(1+D^2)^{-\frac{1}{2}}$. Vice versa, any K-homology class
has a representative that arises from a spectral triple through this exact construction \cite{BJ83}.

We recall that in any K-homology class one can find a representative such that $F=F^*$
and $F^2=1$: this will be called a Fredholm module \cite[\S8.2]{GVF01}. A Fredholm
module is called $1$-summable if $[F,\pi(a)]$ is of trace class for all $a\in\A$ \cite{Con94}.

Clearly, if $\HH$ is finite-dimensional, any Fredholm module $(\A,\HH,F,\gamma)$
is $1$-summable and it is also a spectral triple (the resolvent and bounded commutator
conditions are trivially satisfied). There is a pairing between K-homology and K-theory,
that for $1$-summable even Fredholm modules is given by
\begin{equation}\label{eq:pair}
\inner{[F],[p]}=\tfrac{1}{2}\tr_{\HH\otimes\C^n}(\gamma F[F,\pi(p)]) \;,
\end{equation}
where: $p=p^2=p^*\in M_n(\A)$ is a projection, representative of an element $[p]$
in the K-theory group $K_0(\A)$, $[F]$ is the class of the $1$-summable Fredholm module
$(\A,\HH,F,\gamma)$, $\pi$ is extended to a representation of $M_n(\A)$
on $\HH\otimes\C^n$ in the obvious way, and $\tr_{\HH\otimes\C^n}$ is the trace on $\HH\otimes\C^n$.
Using \eqref{eq:pair} any Fredholm module corresponds to a linear map:
\begin{equation}\label{eq:chern}
\mathrm{ch}_F:K_0(\A)\to\C \;,\qquad
\mathrm{ch}_F([p]):=\inner{[F],[p]}  \;,
\end{equation}
usually called Chern-Connes character.

\medskip

Suppose we are interested in Fredholm modules for the algebra $\C$, i.e.~functions on
the space with one point.
For any $z\in\C$ one has $\pi(z)=z\pi(1)$, and if the representation is unital,
{then} $\pi(1)$ commutes with any operator $F$, {so that the} Chern-Connes character
\eqref{eq:chern} is identically zero. To get a {non-zero Chern-Connes character},
we need to use a representation $\pi$ that is not unital.

A non-trivial Fredholm module $(\C,\widetilde{\HH},\widetilde{F},\widetilde{\gamma})$ is given by
$\widetilde{\HH}=\C^2$, with representation, operator $\widetilde{F}$ and grading given by:
$$
\widetilde{\pi}(z)=\mat{z & 0 \\ 0 & 0} \;,\qquad
\widetilde{F}=\mat{0 & 1 \\ 1 & 0} \;,\qquad
\widetilde{\gamma}=\mat{1 & \;0 \\ 0 & \!\!-1} \;.
$$
Note that
$$
\tr_{\C^2}(\widetilde{\gamma}\widetilde{F}[\widetilde{F},\widetilde{\pi}(z)])=2z \;.
$$
Given an element $[p]\in K_0(\C)$, from \eqref{eq:pair} we get
$$
\inner{[\widetilde{F}],[p]}
=\tfrac{1}{2}\tr_{\C^2}(\widetilde{\gamma}\widetilde{F}[\widetilde{F},\textstyle{\sum_i}\widetilde{\pi}(p_{ii})])
=\textstyle{\sum_i}p_{ii} \;.
$$
This is exactly the rank of $p$. It is well known that the above Fredholm module generates
$K^0(\C)\simeq\Z$ (any other Fredholm module is equivalent to a multiple of this one).

\subsection[Pull-backs and products]{Pull-back of Fredholm modules/``amplification'' of spectral triples}
Here we describe two ways to construct Fredholm modules or spectral triples on an algebra $\A$
using the basic Fredholm module $(\C,\widetilde{\HH},\widetilde{F},\widetilde{\gamma})$ of
previous section. These will be applied then to the study of the algebra $\C^2$.

\subsubsection{Pull-back}\label{sec:pull}

Given a connected locally compact Hausdorff space $X$, we can use an irreducible representation ---
i.e.~a map $C_0(X)\to\C$, $f\mapsto f(x)$, with $x\in X$ --- to obtain a Fredholm
module over $C_0(X)$ as a pullback of the Fredholm module over $\C$ given above:
the corresponding Chern-Connes map, evaluated on a projection, gives the rank
of the corresponding vector bundle. This is true even for some quantum spaces,
for example quantum complex projective spaces \cite{DL09a}.
For $X$ compact, one of the generators of $K_0(C(X))$ is the trivial projection $p=1$
(the constant function); to get a full set of generators of $K^0(C(X))$ one is forced
to use the construction above, as any Fredholm module with a unital representation
will have a trivial pairing with $p=1$.

\smallskip

More generally if $\A$ is any associative involutive complex algebra, one can use a
one-dimensional irreducible representation $\chi:\A\to\C$ (if any) to pull-back the Fredholm module
$(\C,\widetilde{\HH},\widetilde{F},\widetilde{\gamma})$.
The result is a Fredholm module $(\A,\widetilde{\HH},\widetilde{F},\widetilde{\gamma})$
with $\widetilde{\HH}$, $\widetilde{F}$, and $\widetilde{\gamma}$ as above, and with
representation of $\A$ on $\widetilde{\HH}$ given by $\widetilde{\pi}\circ\chi$:
$$
\widetilde{\pi}\circ\chi(a)=\mat{\chi(a) & 0 \\ 0 & 0} \;,\qquad
\forall\;a\in\A\;.
$$
This is what we do, for example, for quantum complex projective spaces or for
the standard Podle\'s sphere, to get the last generator of the K-homology \cite{DL09a}.

\subsubsection{Amplification}\label{sec:amp}

As explained, the Fredholm module $(\C,\widetilde{\HH},\widetilde{F},\widetilde{\gamma})$
above is also an even spectral triple (over the space with one point), and given any other
spectral triple $(\A,\pi_0,\HH_0,D_0)$ we can form their product, that we denote by $(\A,\pi,\HH,D)$
(we explicitly indicate the representation symbols).
Clearly we are not changing the algebra: $\A\otimes\C\simeq\A$. On the other hand,
$\HH=\HH_0\otimes\C^2$, and the representation and Dirac operator are given by:
$$
\pi(a)=\mat{\pi_0(a) & 0 \\ 0 & 0} \;,\qquad
D=\mat{D_0 & 1 \\ 1 & \!\!-D_0 } \;.
$$
If the former spectral triple is even, with grading $\gamma_0$, the latter is even too, with
grading
$$
\gamma=\mat{\gamma_0 & 0 \\ \;0 & \!\!-\gamma_0} \;.
$$
The advantage is that one can start from a spectral triple that is trivial in K-homology,
and get a new spectral triple with a non-trivial K-homology class.

\subsection{K-homology of $\C^2$}\label{sec:Ctwo}
The K-theory and K-homology of $\C^2=\C\oplus\C$ are well known (here we denote elements as
pairs $(a,b)$, instead of writing $a\oplus b$).
We know that the K-theory is $K_0(\C\oplus\C)\simeq\Z\oplus\Z$ (see e.g.~Exercise 6.I(h) of \cite{WO93}),
with generators given by
$$
p_+:=(1,0)
 \;,\qquad
p_-:=(0,1)
 \;.
$$
(We are not interested in $K_1$, that in this example is zero anyway.)

\smallskip

Concerning K-homology, using the two irreducible representations of $\C^2$, given
by the two pure states $\varphi_+(a,b)=a$ and $\varphi_-(a,b)=b$, we
can get two Fredholm modules $(\C^2,\HH_+,F_+,\gamma_+)$
and $(\C^2,\HH_-,F_-,\gamma_-)$, as explained in \S\ref{sec:pull}.
We have $\HH_+=\HH_-=\C^2$,
$$
F_+=F_-=\mat{0 & 1 \\ 1 & 0} \;,\qquad
\gamma_+=\gamma_-=\mat{1 & \;0 \\ 0 & \!\!-1} \;,
$$
and the only difference is in the representation
$$
\pi_+(a,b)=\mat{a & 0 \\ 0 & 0} \;,\qquad
\pi_-(a,b)=\mat{b & 0 \\ 0 & 0} \;.
$$
Using \eqref{eq:chern} one checks
that the pairing with K-theory is $\inner{[F_i],[p_j]}=\delta_{ij}$,
proving that we have a dual pair of generators of K-theory and K-homology.

\smallskip

For metric purposes, these two Fredholm modules are not very interesting
since the corresponding spectral distance between pure states is infinite (in
both cases we have elements not proportional to $1$ commuting with the `Dirac'
operator $F$: those in the kernel of the representation). We now describe two
spectral triples that are more suitable for metric purposes.

The first unital spectral triple $(\A_1,\pi_1,\HH_1,D_1,\gamma_1)$ is given by
$\A_1=\HH_1=\C^2$, with
$$
\pi_1(a,b)=\mat{a & 0 \\ 0 & b} \;,\qquad
D_1=\frac{1}{\lambda}\,\mat{0 & 1 \\ 1 & 0} \;,\qquad
\gamma_1=\mat{1 & 0 \\ 0 & -1} \;,
$$
with $\lambda>0$ a fixed scale.
The distance between the two pure states of the algebra is easily computed,
and given by (see e.g.~page 35 of \cite{Con94}):
\begin{equation}\label{eq:lambdad}
d_{\A_1,D_1}(\varphi_0,\varphi_1)=\lambda \;.
\end{equation}
The corresponding Fredholm module is given by $F_1=\lambda D_1$.

Another even Fredholm module $(\A_2,\pi_2,\HH_2,F_2,\gamma_2)$
is obtained as the ``amplification'' of $(\A_1,\pi_1,\HH_1,0,I_2)$ (note that the grading $I_2$
anticommutes with the zero Dirac operator), as explained in
\S\ref{sec:amp}. The result is $\HH_2=\C^4$, with grading $\gamma_2=\mathrm{diag}(1,1,-1,-1)$,
and with (non-unital) $*$-representation $\pi_2:\C^2\to M_4(\C)$ and $F_2$ given by:
$$
\pi_2(a,b)=\left[\begin{array}{cccc}
a & 0 & 0 & 0 \\ 0 & b & 0 & 0 \\ 0 & 0 & 0 & 0 \\ 0 & 0 & 0 & 0
\end{array}\right] \;,\qquad\quad
F_2=\left[\begin{array}{cccc}
0 & 0 & 1 & 0 \\ 0 & 0 & 0 & 1 \\ 1 & 0 & 0 & 0 \\ 0 & 1 & 0 & 0
\end{array}\right] \;.
$$
For the Dirac operator, it is convenient to choose the normalization $D_2=2\mu^{-1}F_2$, with
$\mu>0$ a fixed length scale.
One easily checks that $\|[D_2,\pi_2(a,b)]\|=2\mu^{-1}\max\{|a|,|b|\}$.
From this, it follows that
\begin{equation}\label{eq:mud}
d_{\A_2,D_2}(\varphi_0,\varphi_1)=\mu \;.
\end{equation}
The Dirac operators $D_1$ and $D_2$ correspond to geometries that are ``topologically''
inequivalent, i.e.~to different classes in $K^0(\C^2)$. More precisely, computing the pairing
with the projections $p_+$ and $p_-$ one proves the following relations with the generators of $K^0(\C^2)$:
$$
[F_1]=[F_+]-[F_-] \;,\qquad
[F_2]=[F_+]+[F_-] \;.
$$

%%% ======================================================================
\section{The importance of being non-degenerate}\label{sec:deg}
The proof of \eqref{eq:mainB} works only for unital spectral triples.
In \S\ref{sec:4} and \S\ref{sec:5} we show what happens if one of the two spectral triples
is not unital: in the former case the algebra is unital but the representation
is not, in the latter case the algebra is itself non-unital.
In both cases the inequality \eqref{eq:mainB} is not true (it is violated
already by pure states).

\subsection{A product of two-point spaces}\label{sec:4}
Here we consider the the product $(\A,\pi,\HH,D)$  of the spectral triples
$(\A_1,\pi_1,\HH_1,D_1,\gamma_1)$ and
$(\A_2,\pi_2,\HH_2,D_2)$ over the algebra $\A_1=\A_2=\C^2$ introduced
in \S\ref{sec:Ctwo}. From \eqref{eq:lambdad} and \eqref{eq:mud} we have:
$$
d_{\A_1,D_1}(\varphi_+,\varphi_-)=\lambda \;,\qquad
d_{\A_2,D_2}(\varphi_+,\varphi_-)=\mu \;,
$$
where $\varphi_+(a,b)=a$ and $\varphi_-(a,b)=b$ are the two pure states of $\C^2$.
If \eqref{eq:mainA} were true in the non-unital case, we would expect
$d_{\A,D}(\varphi_+\otimes\varphi_+,\varphi_-\otimes\varphi_-)\geq\sqrt{\lambda^2+\mu^2}$.
The next proposition shows that this is not the case.

\begin{prop}\label{prop:indep}
The distance $d_{\A,D}(\varphi_+\otimes\varphi_+,\varphi_-\otimes\varphi_-)=\mu$
is independent of $\lambda$.
\end{prop}
\begin{proof}
Recall that $\A=\A_1\otimes\A_2$, $\HH=\HH_1\otimes\HH_2$ and $D=D_1\otimes1+\gamma_1\otimes D_2$.
If $\{e_i\}_{i=1}^n$ is the canonical orthonormal basis of $\C^n$, a unitary map
$U:\HH=\C^2\otimes\C^4\to\C^8$ is defined by
$U(e_1\otimes e_i)=e_i$ and $U(e_2\otimes e_1)=e_{i+4} \;\forall\;i=1,\ldots,4$.

An isomorphism $\imath:\A=\C^2\otimes\C^2\to\C^4$ is given by
$$
\imath\big((a_1,a_2)\otimes (b_1,b_2)\big):=(a_1b_1,a_1b_2,a_2b_1,a_2b_2) \;.
$$
The states $\varphi_+\otimes\varphi_+$ and $\varphi_-\otimes\varphi_-$ can be pulled-back
to states on $\C^4$ given by
$$
(\varphi_+\otimes\varphi_+)\circ\imath^{-1}(a_1,\ldots,a_4)=a_1 \;,\qquad
(\varphi_-\otimes\varphi_-)\circ\imath^{-1}(a_1,\ldots,a_4)=a_4 \;.
$$
The representation $\pi_1\otimes\pi_2$ gives the following representation
$\pi(a):=U\big((\pi_1\otimes\pi_2)\imath^{-1}(a)\big)U^*$ of $a=(a_1,\ldots,a_4)\in\C^4$
that is explicitly given by:
$$
\pi(a_1,\ldots,a_4)=\mathrm{diag}(a_1,a_2,0,0,a_3,a_4,0,0) \;,
$$
and the Dirac operator becomes the $8\times 8$ matrix
$$
D=\mat{D_2 & \lambda^{-1}I_4 \\ \lambda^{-1}I_4 & -D_2} \;,
$$
where $I_4$ is the $4\times 4$ identity matrix.
The distance $d_{\A,D}(\varphi_+\otimes\varphi_+,\varphi_-\otimes\varphi_-)$
is then the supremum of $a_1-a_4$ over $a=(a_1,\ldots,a_4)\in\R^4$ (it is enough
to consider self-adjoint elements) with the condition $\|[D,\pi(a)]\|\leq 1$.
Applying the permutation 
$\binom{1\;2\;3\;4\;5\;6\;7\;8}{1\;2\;7\;8\;3\;4\;5\;6}$ to rows and columns
of $[D,\pi(a)]$, one gets the matrix
$$
\mat{0_4 & \!\!-B_a \\ B_a^* & 0_4} \;,\qquad\mathrm{with}\quad
B_a:=
{\frac{1}{\mu}}\left[\begin{array}{cccc}
2a_1 & 0 & {\mu}\lambda^{-1}(a_1-a_3) & 0 \\ 0 & 2a_2 & 0 & {\mu}\lambda^{-1}(a_2-a_4) \\ 0 & 0 & 2a_3 & 0 \\ 0 & 0 & 0 & 2a_4
\end{array}\right] \;.
$$
Since permutation matrices are unitary, $\|[D,\pi(a)]\|=\|B_a\|$.
Denoting by $b_{ij}$ the matrix elements of $B_a$,
since $b_{ij}\leq\|B_a\|$, we have
$$
a_1-a_4=\tfrac{\mu}{2}(b_{11}-b_{44})\leq \mu\|B_a\|= \mu\|[D,\pi(a)]\|  \;,
$$
proving that the distance is no greater than $\mu$.
If $a_1=-a_2=a_3=-a_4={\mu}/2$, then
$B_a=\mathrm{diag}(1,-1,1,-1)$ has norm $1$, so that the distance is no less than $a_1-a_4=\mu$.
\end{proof}

\begin{cor}
For $\lambda/\mu\geq 0$ ($\mu\neq 0$), the ratio
$$
k_\lambda:=
\frac{d_{\A,D}(\varphi_+\otimes\varphi_+,\varphi_-\otimes\varphi_-)}{
\sqrt{d_{\A_1,D_1}(\varphi_+,\varphi_-)^2+d_{\A_2,D_2}(\varphi_+,\varphi_-)^2}}
=\frac{\mu}{\sqrt{\lambda^2+\mu^2}}
$$
assumes all the values in the interval $(0,1]$ (compare this with the
situation in \S\ref{sec:2.2}).
\end{cor}

Similarly to \eqref{eq:sqrtoftwo}, we could think of replacing \eqref{eq:mainB}
with a weaker inequality
\begin{equation}\label{eq:ineqSDaT}
d(\varphi,\varphi')\geq k
\sqrt{d_1(\varphi_1,\varphi'_1)^2+d_2(\varphi_2,\varphi'_2)^2} \;,
\end{equation}
for some $k\geq 0$. Last corollary shows that the only value of $k$ such that
\eqref{eq:ineqSDaT} is valid for all spectral triples is $k=0$.

\smallskip

Observe also that Cor.~\ref{cor9} is no longer valid in the non-unital case.
What spoils the proof is the ``$\geq$'' inequality.
Since what really matter is the ratio $\lambda/\mu$, from now on $\mu=1$.

\begin{prop}
If $\lambda>1$ (and $\mu=1$), then
$
d_{\A,D}(\varphi_+\otimes\varphi_+,\varphi_-\otimes\varphi_+)<d_{\A_1,D_1}(\varphi_+,\varphi_-) .
$
\end{prop}
\begin{proof}
The distance between $\varphi_+\otimes\varphi_+$ and $\varphi_-\otimes\varphi_+$
is the supremum of $a_1-a_3$ (over $a_1,\ldots,a_4\in\R$) with the constraint that
$\|B_a\|\leq 1$. Note that
\begin{equation}\label{eq:aoneathree}
a_1-a_3
=\tfrac{\lambda}{1+\lambda}\big\{(a_1-a_3)+\lambda^{-1}(a_1-a_3)\big\}
\leq \tfrac{\lambda}{1+\lambda}\big\{|a_1|+|a_3|+\lambda^{-1}|a_1-a_3| \big\} \;.
\end{equation}
Every $n\times n$ matrix $B$ satisfies (cf.~equations (2.3.11) and (2.3.12) of \cite{GvL96}):
$$
\|B\|_\infty:=\max_{1\leq i\leq n}\sum\nolimits_{j=1}^n|b_{ij}|\leq \sqrt{n}\|B\|
\quad\mathrm{and}\quad
\|B\|_1:=\max_{1\leq j\leq n}\sum\nolimits_{i=1}^n|b_{ij}|\leq \sqrt{n}\|B\| \;.
$$
In our case $n=4$, and looking at the third column resp.~first row we get:
$$
2|a_3|+\lambda^{-1}|a_1-a_3|\leq \|B_a\|_\infty\leq 2\|B_a\| \;,\qquad
2|a_1|+\lambda^{-1}|a_1-a_3|\leq \|B_a\|_1\leq 2\|B_a\| \;.
$$
Thus
$
\,|a_1|+|a_3|+\lambda^{-1}|a_1-a_3|\leq 2\|B_a\|\,
$
and by \eqref{eq:aoneathree}:
$
\,a_1-a_3\leq \frac{2\lambda}{1+\lambda} \|B_a\|\, .
$
This proves that
$d_{\A,D}(\varphi_+\otimes\varphi_+,\varphi_-\otimes\varphi_+)\leq \frac{2\lambda}{1+\lambda}$.
On the other hand $d_{\A_1,D_1}(\varphi_+,\varphi_-)=\lambda$, and
$$
\frac{2\lambda}{1+\lambda}<\lambda
$$
for every $\lambda>1$. This concludes the proof.
\end{proof}

\begin{rem}
For $\lambda\to\infty$, $d_{\A_1,D_1}(\varphi_+,\varphi_-)$ diverges while
$d_{\A,D}(\varphi_+\otimes\varphi_+,\varphi_-\otimes\varphi_+)\leq 2$.
\end{rem}

%%% ======================================================================
\subsection{Two-sheeted real line}\label{sec:5}
In this example, the first (unital) spectral triple
$(\A_1,\pi_1,\HH_1,D_1,\gamma_1)$ is the one introduced
in \S\ref{sec:Ctwo}, depending on a scale $\lambda>0$.
The second (non-unital) spectral triple
$(\A_2,\pi_2,\HH_2,F_2)$ is obtained as the ``amplification'',
cf.~\S\ref{sec:amp}, of the canonical spectral triple of the real line:
$$
\big(C_0^\infty(\R),L^2(\R),\tfrac{1}{2}\D\big)\;,\qquad\D:=i\tfrac{\de}{\de x} \;.
$$
(The normalization $1/2$ of the Dirac operator allows to simplify some formulas.)
The result is $\A_2=C_0^\infty(\R)$
with representation $\pi_2$ on $\HH_2=L^2(\R)\otimes\C^2$ 
Dirac operator $F_2$ given by:
$$
\pi_2(f)=\mat{f & 0 \\ 0 & 0} \;,\qquad\quad
F_2=\frac{1}{2}\mat{\D & \;2 \\ 2 & \!\!-\D} \;.
$$
where $f$ acts on $L^2(\R)$ by pointwise multiplication.
To get nicer formulas, we prefer to compute the distance using $D_2=2F_2$
rather than $F_2$.

Now, let $(\A,\HH,D)$ be the product of $(\A_1,\pi_1,\HH_1,D_1,\gamma_1)$
and $(\A_2,\pi_2,\HH_2,D_2)$.

\begin{prop}\label{prop20}
For any $x,y\in\R$ and $\lambda>0$, we have
$$
d_{\A,D}(\varphi_+\otimes\delta_x,\varphi_-\otimes\delta_y)
{\leq 1} \leq
\lambda^{-1}\sqrt{d_{\A_1,D_1}(\varphi_+,\varphi_-)^2+d_{\A_2,D_2}(\delta_x,\delta_y)^2} \;.
$$
Notice that for $\lambda\to\infty$, $d_{\A_1,D_1}(\varphi_+,\varphi_-)$ diverges while
$d_{\A,D}(\varphi_+\otimes\delta_x,\varphi_-\otimes\delta_x)$ is never greater than $1$. In particular, if $\lambda>1$ one has
$$
d_{\A,D}(\varphi_+\otimes\delta_x,\varphi_-\otimes\delta_x)\neq
d_{\A_1,D_1}(\varphi_+,\varphi_-) \;.
$$
\end{prop}
\begin{proof}
Recall that $\A=\A_1\otimes\A_2$, $\HH=\HH_1\otimes\HH_2$ and $D=D_1\otimes1+\gamma_1\otimes D_2$.

We identify $\C^2\otimes\HH_2$ with $L^2(\R)\otimes\C^4$. The representation of
$(f_+,f_-)\in\C^2\otimes C_0(\R)$ is
$\pi(f_+,f_-)=\mathrm{diag}(f_+,0,f_-,0)$,
and the Dirac operator is
$$
D=\mat{D_2 & \lambda^{-1}I_2 \\ \lambda^{-1}I_2 & -D_2}
=\left[\begin{array}{cccc}
\D & 2 & \;\;\lambda^{-1}\!\! & 0 \\
2 & \!\!\!-\D\; & 0 & \;\;\lambda^{-1}\!\! \\
\;\;\lambda^{-1}\!\! & 0 & \!\!\!-\D\; & \!\!\!-2\; \\
0 & \;\;\lambda^{-1}\!\! & \!\!\!-2\; & \D
\end{array}\right] \;,
$$
where $I_2$ is the $2\times 2$ identity matrix.
We have
$$
[D,\pi(f_+,f_-)]
=\left[\begin{array}{cccc}
if_+ ' & -2f_+ & -\lambda^{-1}(f_+-f_-) & 0 \\
2f_+ & 0 & 0 & 0 \\
\lambda^{-1}(f_+-f_-) & 0 & -if_-' & 2f_- \\
0 & 0 & -2f_- & 0
\end{array}\right] \;.
$$
By taking the sup over vectors with only the second component different from
zero we prove that $\|2f_+ \|_\infty\leq\|[D,\pi(f_+,f_-)]\|$.
Similarly $\|2f_- \|_\infty\leq\|[D,\pi(f_+,f_-)]\|$.
But
$$
(\varphi_+\otimes\delta_x)(f_+,f_-)
-(\varphi_-\otimes\delta_y)(f_+,f_-)
=f_+(x)-f_-(y)\leq\|f_+ \|_\infty+\|f_- \|_\infty
\leq\|[D,\pi(f_+,f_-)]\| \;,
$$
and this proves that
$d_{\A,D}(\varphi_+\otimes\delta_x,\varphi_-\otimes\delta_y)\leq 1$.
This proves the first inequality in Prop.~\ref{prop20}.
The other one follows from the simple observation that
$$
d_{\A_1,D_1}(\varphi_+,\varphi_-)^2+d_{\A_2,D_2}(\delta_x,\delta_y)^2
\geq d_{\A_1,D_1}(\varphi_+,\varphi_-)^2=\lambda^2 \;,
$$
last equality being \eqref{eq:lambdad}.
\end{proof}

%%% ======================================================================

\section{Conclusion}
As often advertised by Connes, the ``line element'' in noncommutative
geometry has to be thought as the inverse of the Dirac operator,
$$
\text{`` }\,\de s\sim D^{-1}\text{ ''}.
$$
For a product of two spectral triples $X_1=(\A_1,\HH_1,D_1)$ and $X_2=(\A_2,\HH_2,D_2)$,
noticing that
\begin{equation}
D^2 = (D_1 \otimes \Gamma_2 + 1\otimes D_1)^2 = D_1^2 \otimes 1 + 1
\otimes D_2^2,
\label{eq:19}
\end{equation}
this yields a ``inverse Pythagoras equality'' \cite{Conblog}
\begin{equation}\label{eq:desminustwo}
\frac{1}{\de s^2}=\frac{1}{\de s_1^2}+\frac{1}{\de s_2^2} \;.
\end{equation}
In \cite{Martinetti:2009kx}, we discussed why it was possible to invert
\eqref{eq:desminustwo} in case of the product of a manifold by $\C^2$, so that to retrieve Pythagoras
theorem.
The main result of this paper is to show that for the product of arbitrary (unital) 
spectral triples, \eqref{eq:desminustwo} can be ``integrated'' and leads
to the inequalities \eqref{eq:Pineq} for the spectral distance. More precisely, if
$\varphi:=\varphi_1\otimes \varphi_2$ and $\varphi'=\varphi_1'\otimes \varphi_2'$ are arbitrary
separable states on $\A=\A_1\otimes\A_2$ (not necessarily pure), and
$d$ (resp.~$d_i$) is the spectral distance of $X$ (resp.~$X_i$) then
the inequalitied \eqref{eq:Pineq} hold:
$$
\sqrt{d_1(\varphi_1,\varphi'_1)^2+d_2(\varphi_2,\varphi'_2)^2}\leq
d(\varphi,\varphi') \leq
\sqrt{2}\sqrt{d_1(\varphi_1,\varphi'_1)^2+d_2(\varphi_2,\varphi'_2)^2} \;.
$$
These inequalities already appeared in \cite[Prop.~II.4]{MT11}, where one of the two spaces
was assumed to be the two-point space $\C^2$ and only pure states were considered. Here we proved them in full generality: \eqref{eq:ineqSDb} holds for arbitrary spectral triples and \eqref{eq:ineqSDa} holds if the spectral
triples $X_1$ and $X_2$ are both unital.

In \S4 we provide two elementary (commutative) examples, where one of the two
spectral triples is non-unital, and we show that the inequality \eqref{eq:ineqSDa} is violated even by pure states.

In the case of the Wasserstein distance, we argued that Pythagoras equality is a pro\-perty of pure states, and
it does not hold if we consider non-pure states. Besides the product of two manifolds, it is
not clear whether the purity of states is an essential conditions to
retrieve Pythagoras theorem: for the product of a manifold by
$\C^2$, the distance between two separable non-pure states $\varphi =
\varphi_1\otimes\varphi_2$, $\varphi' =
\varphi'_1\otimes\varphi'_2$ is still
unknown, except when $\varphi_i = \varphi'_i$ for either $i=1$ or
$i=2$. Then Pythagoras equality
is trivially satisfied.  For the  product of the Moyal plane by
$\C^2$, the purity of the states does not seem to be a relevant
criterion: Pythagoras theorem holds true for translated states, pure or not, and
we do not know whether it holds for arbitrary pure states. In any
case, it would be interesting to find a nice class of noncommutative spectral triples
where Pythagoras equality can be proved for pure states.

%%% ======================================================================

\end{document}